\documentclass[runningheads]{llncs}
\usepackage{graphicx} 
\usepackage{complexity}
\usepackage{todonotes}
\usepackage{amsfonts, 
amsmath, amssymb, mathrsfs}

\usepackage{tikz}
\usetikzlibrary{calc, arrows.meta}

\usepackage{hyperref}
\usepackage{authblk}
\usepackage{yfonts}
\usepackage{xspace}
\usepackage{comment}
\usepackage{subcaption}

\newcommand{\seq}{\mathfrak{S}\xspace}
\newcommand{\classed}[1]{{\sc {#1}-Edge Deletion}\xspace}
\newcommand{\classvs}[1]{{\sc {#1}-Vertex Splitting}\xspace}
\newcommand{\Yes}{\textsf{YES}\xspace}

\newcommand{\cat}{\mathcal{C}at\xspace}
\newcommand{\up}{\mathcal{P}\xspace}

\newcommand{\nacim}[1]{{\color{purple}N: #1}}

\begin{document}

\title{On the Complexity of Vertex-Splitting Into an Interval Graph\thanks{This research project was supported by the Lebanese American University under the President’s Intramural Research Fund PIRF0056.}
}

\author{Faisal N. Abu-Khzam\inst1
\and Dipayan Chakraborty\inst1\inst2 
\and Lucas Isenmann\inst1\inst3
\and Nacim Oijid\inst1\inst4\thanks{The fourth author was partly supported by the Kempe Foundation Grant No. JCSMK24-515 (Sweden).}}
\authorrunning{Abu-Khzam, Chakraborty, Isenmann and Oijid}
\institute{
Lebanese American University, 
Beirut, Lebanon.
\and
Centrale M\'{e}diterran\'{e}e, LIS, Aix-Marseille Universit\'{e},
Marseille, France
\and 
Université de Strasbourg, France
\and
Umeå University, Umeå, Sweden}

\maketitle 
 
\thispagestyle{empty}

\begin{abstract}
Vertex splitting is a graph modification operation in which a vertex is replaced by multiple vertices such that the union of their neighborhoods equals the neighborhood of the original vertex.
We introduce and study vertex splitting as a graph modification operation for transforming graphs into interval graphs.
Given a graph $G$ and an integer $k$, we consider the problem of deciding whether $G$ can be transformed into an interval graph using at most $k$ vertex splits.
We prove that this problem is \NP-hard, even when the input is restricted to subcubic planar bipartite graphs. We further observe that vertex splitting differs fundamentally from vertex and edge deletions as graph modification operations when the objective is to obtain a chordal graph, even for graphs with maximum independent set size at most two.
On the positive side, we 
give a polynomial-time algorithm for transforming, via a minimum number of vertex splits, a given graph into a disjoint union of paths, and that splitting triangle free graphs into unit interval graphs is also solvable in polynomial time.



\end{abstract}

\section{Introduction}

Interval graphs form a fundamental subclass of intersection graphs and arise naturally in a wide range of applications, including scheduling, genetics, and computational biology
\cite{benzer1959reading,booth1976testing,golumbic1980algorithmic,iwata2004interval,tucker1972structure,zhang2007using}.
An interval graph is the intersection graph of intervals on the real line. This class admits strong structural characterizations and supports efficient algorithms: recognizing interval graphs and constructing interval representations can be done in linear time \cite{booth1976testing}.

From an algorithmic standpoint, interval graphs are appealing because of their one-dimensional structure, which enables efficient solutions to many graph problems that are computationally intractable on general graphs.
Once an interval representation is available, problems such as graph coloring, maximum clique, maximum independent set, and minimum vertex cover admit linear-time algorithms \cite{golumbic1980algorithmic}, and efficient solutions are also known for Hamiltonian path, domination variants, and a variety of scheduling and resource allocation problems \cite{booth1976testing,tucker1972structure}. At the same time, the very same rigid structure makes interval graphs difficult to obtain from arbitrary graphs.
Classical graph modification problems such as Vertex Deletion or Edge Deletion into interval graphs, unit interval graphs, or chordal graphs are known to be \NP-complete, even under strong restrictions on the input graph \cite{lewis1980node,natanzon2001complexity,van2010measuring}.

In this work, we investigate an alternative modification mechanism based on \emph{vertex splitting}.
Given a graph $G=(V,E)$, a vertex split replaces a vertex $v\in V$ by two new vertices $v_1$ and $v_2$, each inheriting a (possibly disjoint) subset of the neighbors of $v$.
Rather than removing vertices or edges, vertex splitting redistributes adjacency, allowing a single vertex to be represented by multiple intervals.
This operation preserves information while relaxing local structural constraints.

We introduce and study the following problem: given a graph $G$ and an integer $k$, can $G$ be transformed into an interval graph using at most $k$ vertex splits?
We believe this problem provides a finer-grained notion of distance to the class of interval graphs and differs fundamentally from deletion-based modification problems.
In particular, we show that the number of vertex splits required to transform a graph into a chordal graph can differ from the number of vertex deletions required to achieve the same goal, even when the input graph has independence number 2 (i.e., the size of a maximum independent set is 2).
This separation further demonstrates that known results for vertex deletion into chordal or interval graphs do not directly transfer to vertex splitting, and that splitting-based modification requires separate analysis. 

We should mention that transforming a graph into an interval graph via a sequence of $k$ vertex splitting operations can be viewed as a transformation/modification into a multi-intersection graph of intervals, which differs slightly from the representation as a k-Interval Graph problem defined in \cite{trotter1979}. In this latter problem the objective is to represent each vertex by $k$ intervals....

Vertex splitting has been studied in several other contexts, including planarization and graph thickness \cite{eppstein2018planar,nollenburg2025planarizing}, rigidity theory \cite{whiteley1990vertex}, flow and decomposition problems \cite{zhang2002circular}, and clustering and graph editing, where it enables the modeling of overlapping structures \cite{abu2025complexity,abu2018cluster,iwoca,firbas2025complexity}.
However, its role as a mechanism for enforcing classical intersection graph structure has remained largely unexplored.

Our contributions can be summarized as follows.
We initiate a systematic study of vertex splitting as a graph modification operation toward interval graphs.
We prove that transforming a graph into an interval graph using a bounded number of vertex splits is \NP-hard, even under strong restrictions on the input.
We show that vertex splitting exhibits complexity behavior that differs fundamentally from vertex deletion, including a strict separation between chordal vertex deletion and chordal vertex splitting.
At the same time, we identify tractable cases: we show that splitting any graph into a disjoint union of paths can be decided in polynomial time. In fact, we show that splitting a given subcubic triangle free graphs into an interval graph is \NP-hard but splitting triangle free graphs into unit interval graphs is solvable in polynomial time.

\section{Preliminaries}

Let $G = (V,E)$ be a graph and let $v$ be a vertex of $G$. Throughout this article, we use standard graph-theoretic notation and assume all considered graph are simple, unweighted and undirected. The sets $V (= V(G))$ and $E (= E(G))$ denote the vertex set and the edge set of $G$, respectively. We denote by $N(v)$ the \emph{open neighborhood} of a vertex $v$. The chromatic number of $G$ is denoted by $\chi(G)$, and $\overline{G}$ denotes the complement of $G$, i.e., $V(\overline{G}) = V(G)$ and $E(\overline{G}) = \{ (xy) \in V(\overline{G}) | x \neq y \text{ and } (xy) \notin E(G) \}$.

When no ambiguity arises, we identify a vertex subset $S \subseteq V(G)$ with the subgraph induced by $S$, and write $\chi(S)$ instead of $\chi(G[S])$. We denote by $\alpha(G)$ the size of a maximum independent set of $G$.

\medskip
\noindent
\textbf{Chordal, Interval and Unit Interval graphs.}
A graph is said to be \emph{chordal} if it contains no induced cycle of length at least $4$.
A graph is an \emph{interval graph} if it is the intersection graph of intervals on the real line.
A graph is a \emph{unit interval graph} if it is the intersection graph of unit-length intervals; such graphs are also known as \emph{indifference graphs}. It is well-known that an interval graph is chordal. Furthermore, a classical result states that a graph is a unit interval graph if and only if it is chordal and contains no induced claw, net, or tent \cite{wegner1967eigenschaften} (see Figure~\ref{fig:unit_interval_forbidden}).

\vspace{5pt}

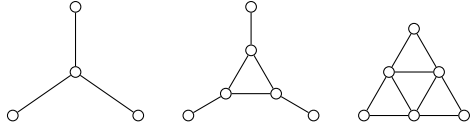
\begin{figure}[htb!]
    \centering
    \begin{tikzpicture}
    [ node_style/.style={circle, draw, inner sep=0.05cm}  ]
        \def\fscale{3} 
    	 \node[node_style] (0) at ($\fscale*(-1, -0.241)$) {};
    	 \node[node_style] (1) at ($\fscale*(-0.722, -0.048)$) {};
    	 \node[node_style] (2) at ($\fscale*(-0.722, 0.241)$) {};
    	 \node[node_style] (3) at ($\fscale*(-0.444, -0.241)$) {};
    	 \node[node_style] (4) at ($\fscale*(-0.222, -0.241)$) {};
    	 \node[node_style] (5) at ($\fscale*(-0.056, -0.144)$) {};
    	 \node[node_style] (6) at ($\fscale*(0.167, -0.144)$) {};
    	 \node[node_style] (7) at ($\fscale*(0.333, -0.241)$) {};
    	 \node[node_style] (8) at ($\fscale*(0.056, 0.048)$) {};
    	 \node[node_style] (9) at ($\fscale*(0.056, 0.241)$) {};
    	 \node[node_style] (10) at ($\fscale*(0.556, -0.241)$) {};
    	 \node[node_style] (11) at ($\fscale*(0.667, -0.048)$) {};
    	 \node[node_style] (12) at ($\fscale*(0.778, -0.241)$) {};
    	 \node[node_style] (13) at ($\fscale*(0.889, -0.048)$) {};
    	 \node[node_style] (14) at ($\fscale*(0.778, 0.144)$) {};
    	 \node[node_style] (15) at ($\fscale*(1, -0.241)$) {};
    
    	 \draw (0) to  (1);
    	 \draw (1) to  (2);
    	 \draw (1) to  (3);
    	 \draw (4) to  (5);
    	 \draw (5) to  (6);
    	 \draw (6) to  (7);
    	 \draw (5) to  (8);
    	 \draw (8) to  (6);
    	 \draw (8) to  (9);
    	 \draw (10) to  (11);
    	 \draw (11) to  (12);
    	 \draw (13) to  (14);
    	 \draw (14) to  (11);
    	 \draw (11) to  (13);
    	 \draw (13) to  (15);
    	 \draw (12) to  (13);
    	 \draw (15) to  (12);
    	 \draw (12) to  (10);
    \end{tikzpicture}
    \caption{The claw, net and tent graphs.}
    \label{fig:unit_interval_forbidden}
\end{figure}

We say that a vertex $v$ is \emph{simplicial} if $G[N(v)]$ is a clique.
A classical characterization, due to Dirac, states that a graph is chordal if and only if it admits a \emph{perfect elimination ordering}, that is, an ordering $v_1,\ldots,v_n$ of the vertices such that $v_i$ is simplicial in $G[\{v_i,\ldots,v_n\}]$ for every $i$ \cite{dirac1961rigid,golumbic1980algorithmic}.

A $(t,d)$-$star$ is a graph consisting of a central vertex $v$, the star center, and $t$ vertex-disjoint paths of length $d$ each that have $v$ as endpoint (each). Figure \ref{fig:t2} shows a $(3,2)$-star, which we refer to hereafter (for simplicity) as $T_2$, also known as a subdivided claw.

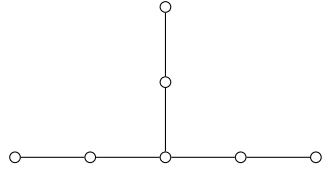
\begin{figure}[!h]
    \centering
\begin{tikzpicture}
    [ node_style/.style={circle, draw, inner sep=0.05cm}  ]
    \def\fscale{2} 
	 \node[node_style] (0) at ($\fscale*(0, 0.5)$) {};
	 \node[node_style] (1) at ($\fscale*(0, 0)$) {};
	 \node[node_style] (2) at ($\fscale*(0, -0.5)$) {};
	 \node[node_style] (3) at ($\fscale*(0.5, -0.5)$) {};
	 \node[node_style] (4) at ($\fscale*(1, -0.5)$) {};
	 \node[node_style] (5) at ($\fscale*(-0.5, -0.5)$) {};
	 \node[node_style] (6) at ($\fscale*(-1, -0.5)$) {};

	 \draw (6) -- (5) -- (2) -- (3) --  (4);
	 \draw (0) -- (1) --  (2);
\end{tikzpicture}
    \caption{The graph $T_2$.}
    \label{fig:t2}
\end{figure}


A \emph{tree} is a connected cycle-free graph. A \emph{caterpillar} is a tree such that removing all its degree-one vertices (i.e. leaves) yields a path, called the spine. Let $\cat$ be the class of caterpillars. For $k$ an integer, let $\cat_k$ be the set of caterpillar graphs of maximum degree $k$. The following characterization will be used in the proof of our main result.

\begin{lemma}\label{lemma: t2free = caterpillar}
Let $T$ be a tree. Then $T$ is $T_2$-free if and only if $T$ is a caterpillar.
\end{lemma}

\begin{proof}
Every caterpillar is trivially $T_2$-free.
Now suppose $T$ is not a caterpillar, then there exists a vertex $v_0$ adjacent to three vertices of degree at least $2$. Since $T$ is a tree, there are no edges or common neighbors other than $v_0$ between these three vertices. Therefore $v_0$ and its three neighbors along with their neighbors (at distance $2$ from $v_0$) induce a subgraph that contains a $T_2$.
\end{proof}

\medskip
\noindent
\textbf{Vertex splitting.}
We define a \emph{(vertex) splitting} of a vertex $v \in V(G)$ as a graph modification operation in which $v$ is replaced by two new vertices $v_1$ and $v_2$, such that every neighbor of $v$ becomes adjacent to $v_1$, or to $v_2$, or to both.
Formally, splitting $v \in V(G)$ yields a graph $G'$ defined as follows. Let $A,B \subseteq N_G(v)$ with $A \cup B = N_G(v)$. Then
\begin{align*}
V(G') &= \bigl(V(G) \setminus \{v\}\bigr) \cup \{v_1, v_2\},\\
E(G') &= \bigl(E(G) \setminus \{vx : x \in N_G(v)\}\bigr)
         \cup \{v_1x : x \in A\} \cup \{v_2x : x \in B\}.
\end{align*}

Each vertex splitting increases the number of vertices by one. The particular case where $A\cap B = \emptyset$ is of interest, at least from a practical standpoint where splitting is supposed to result in clusters, for example \cite{abu2018cluster}. This will be referred to as \emph{exclusive} vertex splitting in the sequel.

\medskip
\noindent
\textbf{Chordal, Interval and Unit Interval Vertex Splitting.}
Let $G$ be a graph. We denote by $ChVS(G)$ the minimum length of a split sequence which turns $G$ into a chordal graph.
Similarly we define $ChVXS(G)$ as the minimum length of an exclusive split sequence which turns $G$ into a chordal graph.
Similarly we define $IVS(G)$, $IVXS(G)$, $UIVS(G)$ and $UIVXS(G)$ for interval and unit interval graphs.

\section{Complexity of Vertex Splitting into Interval graphs}

We prove the \NP-completeness of deciding whether a graph can be transformed into an interval graph using at most a given number 
of vertex splits. We first focus on the exclusive variant of vertex splitting and show its \NP-hardness.

Since inclusive and exclusive splitting are not equivalent in general, hardness for one variant does not automatically transfer to the other.
We next show that, on triangle-free graphs, any sequence of inclusive splits leading to an interval (or unit interval, or chordal) graph can be converted into a sequence of exclusive splits of no greater length.
This assertion, formalized in Lemma~\ref{lemma:inclusive_to_exclusive}, allows us later to deduce hardness results for the inclusive 
variant.

\begin{lemma} \label{lemma:inclusive_to_exclusive}
Let $G$ be a triangle free graph and let $\sigma$ be a sequence of $l$ splits resulting in an interval graph (respectively, a unit interval graph, or a chordal graph).
Then there exists a sequence of at most $l$ exclusive splits resulting in a graph of the same class.
\end{lemma}

\begin{proof}
    
Let us proceed by induction on the number of non exclusive splits.
If all splits of $\sigma$ are exclusive then we are done.
Otherwise let $k$ be the number of non exclusive splits of $\sigma$ and suppose that the statement is true for sequences with less than $k-1$ non exclusive splits.

Consider the last non exclusive split $s_i = (v,A,B)$ of the sequence $\sigma = (s_1, ..., s_l)$ where $A \subseteq N(v)$ and $B \subseteq N(v)$ such that $A \cup B = N(v)$ and $A \cap B \neq \emptyset$. Let $v_1$ and $v_2$ be the two created vertices such that $N(v_1) = A$ and $N(v_2) = B$.

For a set $X$ of vertices of $G$ during the split sequence $\sigma$, $D(X)$ denotes the set $X$ and the the descendants of each of the vertices of $X$, i.e. the set of vertices resulting from the split sequence of the vertices in $X$.
We consider the split $s'_i = (v, A, B \setminus A)$.
For every split $s_j = (v_j, A_j, B_j)$ for $j > i$, we define $s'_j = (v_j, A'_j, B'_j)$ as follows: if $v_j \in D(A \cap B)$, $A'_j = A_j \setminus D(v_2)$ and $B'_j = B_j \setminus D(v_2)$, if $v_j \in D(v_2)$, $A'_j = A_j \setminus D(A \cap B)$ and $B'_j = B_j \setminus D(A \cap B)$, otherwise $A'_j = A_j$ and $B'_j = B_j$.
We define the sequence $\sigma' = (s_1, ... s_{i-1}, s_i', ..., s'_l)$.
The splits $s_i', ..., s'_l$ are exclusive.
Therefore the sequence $\sigma'$ has $k-1$ exclusive splits.

The resulting graph $G(\sigma')$ is a subgraph of $G(\sigma)$ where we have removed the edges between $D(v_2)$ and $D(A \cap B)$.

Finally, since triangle free interval graphs are precisely dijoint unions of caterpillars, and any subgraph of a union of caterpillars is also a union of caterpillars, $G(\sigma')$ is an interval graph and $\sigma'$ requires one less non-exclusive split than $\sigma$. We conclude by induction.

For unit interval graphs, the conclusion is the same: the class of triangle free unit interval graphs is monotone (closed under taking arbitrary subgraph) because it is the class of union of paths. Furthermore, the conclusion is the same for chordal graphs: the class of triangle-free chordal graphs is monotone because it is nothing but the class of forests.
\end{proof}

\begin{corollary} \label{corollary:triangle_free_graphs}
    For a triangle-free graph $G$,
    \begin{itemize}
        \item $ChVS(G) =ChVXS(G)$.
        \item $IVS(G) =IVXS(G)$.
        \item $UIVS(G) =UIVXS(G)$.
    \end{itemize}
\end{corollary}
\begin{proof}
    In general, for any graph $G$, we have $ChVS(G) \leq ChVXS(G)$, $IVS(G) \leq IVXS(G)$ and $UIVS(G) \leq UIVXS(G)$.
    Because of the previous Lemma, we deduce the equalities for triangle-free graphs.
\end{proof}

\begin{theorem}\label{thm:nphardness}
Interval Graph Exclusive Vertex Splitting is \NP-complete, even restricted to planar bipartite subcubic graphs.
\end{theorem}

\begin{proof}
The proof is a reduction from the {\sc Hamiltonian Path} problem on cubic planar graphs, which is known to be \NP-complete~\cite{Garey1976}.

Let $G$ be a cubic planar graph on $n$ vertices $v_1, \dots, v_n$, and let $G'$ be the graph obtained by subdividing every edge of $G$ exactly once. Formally, build $G'$ as follows:

\begin{itemize}

\item For each vertex $v_i \in V(G)$, we add a vertex $u_i$ in $V(G')$.
        
\item for each edge $(v_iv_j) \in E(G)$, with $i<j$, we add a vertex $u_{ij}$ in $V(G')$ and the two edges $(u_iu_{ij})$ and $(u_{ij}u_j)$.
    
\end{itemize}

The construction is depicted in Figure~\ref{fig:construction-NP-hardness}

\begin{figure}[h]
\centering

\begin{subfigure}{0.4\textwidth}
\centering
\hspace{-3cm} \begin{tikzpicture}[scale=1,
  every node/.style={circle, draw, inner sep=2pt, label distance=.5pt}]

\node[label=above:$v_1$] (v1) at (0,2) {};
\node[label=right:$v_2$] (v2) at (1.7,1) {};
\node[label=right:$v_3$] (v3) at (1.7,-1) {};
\node[label=below:$v_4$] (v4) at (0,-2) {};
\node[label=left:$v_5$] (v5) at (-1.7,-1) {};
\node[label=left:$v_6$] (v6) at (-1.7,1) {};

\draw (v1)--(v2)--(v3)--(v4)--(v5)--(v6)--(v1);

\draw (v1)--(v3);
\draw (v4)--(v6);
\draw (v2) to[out=100, in=150, looseness=2.5] (v5);

\end{tikzpicture}
\caption{The graph $G$.}
\end{subfigure} 
\begin{subfigure}{0.4\textwidth}
\centering
\begin{tikzpicture}[scale=1,
  every node/.style={circle, draw, inner sep=2pt}]
  \useasboundingbox (-3,-3) rectangle (3,3);

\node[label=above:$u_1$] (u1) at (0,2) {};
\node[label=right:$u_2$] (u2) at (1.7,1) {};
\node[label=right:$u_3$] (u3) at (1.7,-1) {};
\node[label=below:$u_4$] (u4) at (0,-2) {};
\node[label=left:$u_5$] (u5) at (-1.7,-1) {};
\node[label=left:$u_6$] (u6) at (-1.7,1) {};

\node[label=above:$u_{1,2}$] (u12) at (0.85,1.5) {};
\node[label=right:$u_{2,3}$] (u23) at (1.7,0) {};
\node[label=below:$u_{3,4}$] (u34) at (0.85,-1.5) {};
\node[label=below:$u_{4,5}$] (u45) at (-0.85,-1.5) {};
\node[label=left:$u_{5,6}$] (u56) at (-1.7,0) {};
\node[label=above:$u_{6,1}$] (u61) at (-0.85,1.5) {};
\node[label=above:$u_{1,3}$] (u13) at (0.85,.5) {};
\node[label=above:$u_{4,6}$] (u46) at (-0.85,-.5) {};

\draw (u1)--(u12)--(u2);
\draw (u2)--(u23)--(u3);
\draw (u3)--(u34)--(u4);
\draw (u4)--(u45)--(u5);
\draw (u5)--(u56)--(u6);
\draw (u6)--(u61)--(u1);
\draw (u1)--(u13)--(u3);
\draw (u6)--(u46)--(u4);

\path
  (v2) to[out=100, in=150, looseness=2.5]
  node[pos=0.5, circle, draw, inner sep=2pt,
       label=above:$u_{2,5}$] (u25) {}
  (v5);
\draw (u2) to[out=100, in=30, looseness=1.5] (u25);
\draw (u25) to[out=210, in=150, looseness=1.5] (u5);

\end{tikzpicture}

\caption{The resulting graph $G'$.}
\end{subfigure}
\caption{A planar cubic graph $G$ and the resulting graph $G'$ of the reduction.}
\label{fig:construction-NP-hardness}
\end{figure}
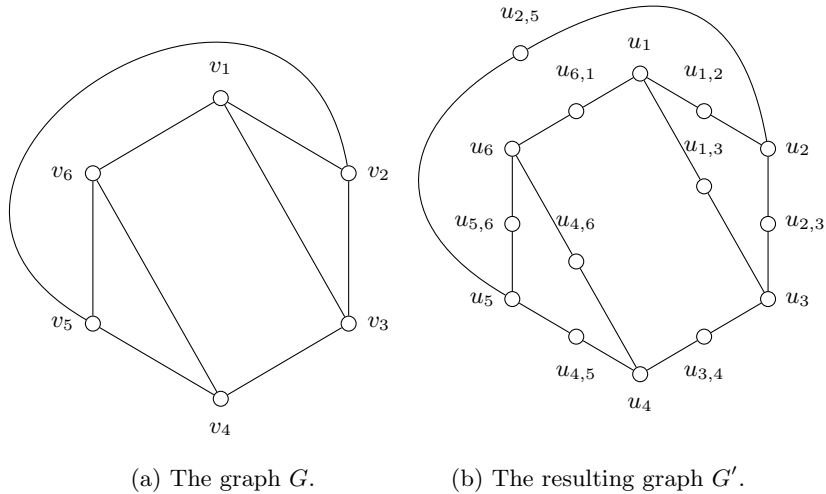
    
Since $G$ is cubic, $|E(G)| = \frac{3}{2}n$, and thus $|V(G')| = \frac{5}{2}n$, and $|E(G')| = 3n$. Note that the vertices of $G'$ can be partitioned into two independent sets $A = \{u_i\}_{1\le i \le n}$ and $B = \{u_{ij}|(v_iv_j) \in G\}$. Thus, $G'$ is planar, bipartite and subcubic.
    
Let $k = \frac{n}{2} +1$. We prove that an Hamiltonian path of $G$ exists if, and only if, $G'$ can be transformed into an interval graphs in $k$ exclusive splits.

Suppose first that there exists a Hamiltonian path in $G$. Up to a relabeling of the vertices of $G$ and $G'$, suppose that this path is $(v_1, v_2, \dots, v_n)$. Consider the following $k$ exclusive splits of $G'$:
For each vertex $u_{i,j}$ with $i \neq j+1$, we split the vertex $u_{i,j}$ into $v_i$ and $v_j$ with $v_i$ adjacent to $u_i$ and $v_j$ adjacent to $u_j$. 
Since $n-1$ vertices $u_{i,j}$ satisfying $j = i+1$ exist, and there is in total $\frac{3}{2}n$ vertices $u_{i,j}$, we performed $\frac{3}{2}n - (n-1) =  \frac{n}{2} + 1$ exclusive splits. The resulting graph is depicted in Figure~\ref{fig:splitting-NP-hardness}
 
By construction, for each $1 < i < n $, exactly one neighbor of $u_i$ has been split, and two neighbors of each of $u_1$ and $u_n$ have been split. Each of these exclusive splits have transformed vertices of degree $2$ into two leaves. Moreover, for $1 \le i \le n$, the vertex $u_{i}$ has not been split. Thus, $\{u_1, u_{1,2}, u_2, \dots, u_{n-1}, u_{n-1,n}, u_n\}$ is a ``dominating'' path of the resulting graph (in the sense that any other vertex has a neighbor in the path). 
Finally, the resulting graph is a caterpillar, which is an interval graph.

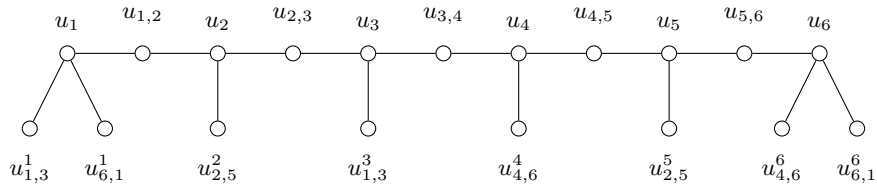
\begin{figure}
    \centering

\begin{tikzpicture}[scale=1,
  every node/.style={circle, draw, inner sep=2pt}]

\node[label=above:$u_1$] (u1) at (0,0) {};
\node[label=above:$u_{1,2}$] (u12) at (1,0) {};
\node[label=above:$u_2$] (u2) at (2,0) {};
\node[label=above:$u_{2,3}$] (u23) at (3,0) {};
\node[label=above:$u_3$] (u3) at (4,0) {};
\node[label=above:$u_{3,4}$] (u34) at (5,0) {};
\node[label=above:$u_4$] (u4) at (6,0) {};
\node[label=above:$u_{4,5}$] (u45) at (7,0) {};
\node[label=above:$u_5$] (u5) at (8,0) {};
\node[label=above:$u_{5,6}$] (u56) at (9,0) {};
\node[label=above:$u_6$] (u6) at (10,0) {};

\node[label=below:$u^1_{1,3}$] (u131) at (-.5,-1) {};
\node[label=below:$u^3_{1,3}$] (u133) at (4,-1) {};
\node[label=below:$u^4_{4,6}$] (u464) at (6,-1) {};
\node[label=below:$u^6_{4,6}$] (u466) at (9.5,-1) {};
\node[label=below:$u^1_{6,1}$] (u161) at (.5,-1) {};
\node[label=below:$u^6_{6,1}$] (u166) at (10.5,-1) {};
\node[label=below:$u^2_{2,5}$] (u252) at (2,-1) {};
\node[label=below:$u^5_{2,5}$] (u255) at (8,-1) {};

\draw (u1)--(u12)--(u2)--(u23)--(u3)--(u34)--(u4)--(u45)--(u5)--(u56)--(u6);

\draw (u131)--(u1)--(u161);
\draw (u466)--(u6)--(u166);

\draw (u252)--(u2);
\draw (u133)--(u3);
\draw (u464)--(u4);
\draw (u255)--(u5);

\end{tikzpicture}
    
    \caption{The resulting graph after the flipping sequence.}
    \label{fig:splitting-NP-hardness}
\end{figure}

Conversely, suppose that $G'$ can be turned into an interval graph $T$ after the application of $k$ exclusive splits.
First, note that, since $G'$ is bipartite, all of its cycles have even length. And observing that any cycle resulting from a split must be an induced cycle of the original graph, we deduce that $T$ has to be a forest.
Moreover, since the exclusive splitting of a vertex does not change the number of edges, but increases the number of vertices by exactly one, and since $|V(G')| + k = |E(G')| +1$, which implies that $|E(T)| = |V(T)| -1$, we conclude that $T$ must be a tree as it contains no cycles. Hence, using Lemma~\ref{lemma: t2free = caterpillar}, $T$ is a caterpillar. 

Now let $P = (w_1, \dots, w_m)$ be a dominating path of the caterpillar.
First, we note that no leaf attached to an internal node of $P$ can be a vertex $u_i$. Indeed, all the neighbors of $u_i$ have degree $2$, and thus, if $u_i$ is a leaf, its neighbor in $P$ can only have one other neighbor. So it is not an internal node of $P$. Up to a change of $P$ into $P'$ by adding the $u_i$s that are leaves attached to the extremities of $P$, we can assume that all the $u_i$s are in $P$. Now, by contracting all the vertices $u_{i,j}$ that are internal vertices of $P$, we obtain a path containing all the $u_i$s, which corresponds to a Hamiltonian path of $G$.
\end{proof}

\begin{corollary}
Splitting into Interval Graphs is \NP-complete even when restricted to planar bipartite subcubic graphs.
\end{corollary}

\begin{proof}
As bipartite graphs are triangle-free graphs, this follows from  Corollary~\ref{corollary:triangle_free_graphs}.
\end{proof}

\section{Splitting into a Disjoint Union of Paths}

Since splitting into a caterpillar graph is already \NP-hard, it seems natural to look at splitting into a disjoint union of paths. In the sequel, we denote by $\up$ the class of graphs consisting of disjoint unions of paths. 

A trail of length $k$ in a graph is a set of $k$ different edges $\{e_1, \dots e_k\}$, such that $|e_i \cap e_{i+1}| = 1$ for all $1 \le i \le k-1$.
For a graph $G$, we denote by $tp(G)$ the minimum number of trails required to partition the edge set of $G$ into disjoint paths. The following Lemma, relating $tp(G)$ to the number of odd-degree vertices in $G$ could be of interest by itself. 

\begin{lemma}\label{lemma trail partition of a graph}
Let $G$ be a connected graph with $2p$ vertices of odd degree. Then $tp(G) = p$ if $p\ge 1$ and $tp(G) = 1$ if $p = 0$.
If $G$ is not connected, then we have $tp(G) = \underset{C \in cc(G)}{\sum} tp(C)$, where $cc(G)$ is the set of connected component of $G$.
\end{lemma}

\begin{proof}
The fact that $tp(G) = \underset{C \in cc(G)}{\sum} tp(C)$, where $cc(G)$ is the set of connected component of $G$ is straightforward. The proof is by induction on $p$. If $p=0$ ($p=1$ resp.), $G$ is Eulerian (semi-Eulerian resp.). Thus, it contains an Eulerian cycle (trail resp.) which is a trail that covers all its edges. Thus, its edges can be covered by a single trail.

If $p\ge 2$, let $u,v$ be two vertices of odd degree in the same connected component $C$ of $G$. If $C$ contains only two vertices of odd degree, we can cover all the edges of $C$ with an eulerian trail $P$. Otherwise, consider any trail $P$ from $u$ to $v$. If removing of the edges forming $P$ disconnects $C$, and we have an Eulerian connected component $C'$ of $C \setminus P$, then we take an Eulerian tour of $C'$ and insert it in $P$ simply as follows: let $x$ be a vertex of $P$ which is in $C'$. We combine the tour starting and ending at $x$ with $P$ to obtain a longer trail. We do the same for every (possible) Eulerian componet of $C\setminus P$. When this process ends, we add $P$ in the partition of the edges.
After that, $G\setminus P$ has $2$ vertices less of odd degree (the endpoints of $P$, since each loses exactly one edge), and no new connected component with only even degree vertices, which provides the result by induction.
\end{proof}

\begin{theorem}\label{thm value splitting into union of paths}
Let $G=(V,E)$ be a graph. The minimum number of splits needed to transform $G$
into a disjoint union of paths is
\[
\sum_{v\in V}\left(\left\lceil \frac{d(v)}{2}\right\rceil-1\right)+r,
\]
where $r$ is the number of connected components of $G$ that contain no vertex of
odd degree.
\end{theorem}

\begin{proof}
Since vertex splitting applies independently to connected components, and $\mathcal{P}$ is closed
under disjoint union, it suffices to prove the statement for a connected graph, and then sum over components. So assume $G$ is connected. Then $r\in\{0,1\}$.

Let $\sigma$ be a split sequence transforming $G$ into $H\in\mathcal{P}$.
Consider an original vertex $v\in V(G)$ and let $D(v)$ be the set of vertices of
$H$ that originate from $v$ (the descendants of $v$ after applying $\sigma$).
Every edge incident to $v$ in $G$ remains present in $H$ and is incident to
exactly one vertex of $D(v)$; therefore
\[
\sum_{x\in D(v)} d_H(x) = d_G(v).
\]
Since $H\in\mathcal{P}$, every vertex of $H$ has degree at most $2$, so
$\sum_{x\in D(v)} d_H(x) \le 2|D(v)|$. Hence $d_G(v)\le 2|D(v)|$, i.e.
$|D(v)|\ge \lceil d_G(v)/2\rceil$.
Each split increases the total number of vertices by $1$, and $|D(v)|=1$ before
any split on descendants of $v$. Thus at least $|D(v)|-1\ge \lceil d_G(v)/2\rceil-1$
splits are necessary “because of $v$”. Summing over all $v\in V(G)$ yields the
unconditional lower bound
\[
|\sigma|\ \ge\ \sum_{v\in V(G)}\left(\left\lceil\frac{d_G(v)}{2}\right\rceil-1\right).
\tag{$\ast$}
\]

If $r=0$, we are done with the lower bound. Assume now $r=1$, i.e. $G$ has no odd-degree vertices (so $G$ is Eulerian).
If equality holds in $(\ast)$, then for every $v$ we must have
$|D(v)|=\lceil d_G(v)/2\rceil$ and every descendant has degree exactly $2$
(because the degree sum over $D(v)$ is exactly $d_G(v)$ and the maximum is $2$).
Consequently, every vertex of $H$ has degree exactly $2$, so each connected
component of $H$ is a cycle. In particular, $H\notin \mathcal{P}$ unless $E=\emptyset$.
Therefore, in the Eulerian connected case, at least one additional split is
necessary beyond $(\ast)$, giving the lower bound
\[
|\sigma|\ \ge\ \sum_{v\in V(G)}\left(\left\lceil\frac{d_G(v)}{2}\right\rceil-1\right)+1.
\]

To show the upper bound, let $T=\{P_1,\dots,P_k\}$ be a partition of $E(G)$ into trails (as in the
definition of $tp(G)$). We build a split sequence that realizes these trails as
vertex-disjoint paths.

Fix a vertex $v\in V(G)$ and let the trails of $T$ use $v$ a total of $t(v)$ times,
counting occurrences (so an internal occurrence contributes $2$ incident edges
of the trail at $v$, and an endpoint occurrence contributes $1$ incident edge).
Then the $d_G(v)$ incident edges of $v$ are grouped into $t(v)$ blocks, each block
having size $2$ except possibly one block of size $1$ if $v$ is an endpoint of a
trail. Hence $t(v)\ge \lceil d_G(v)/2\rceil$ for all $v$.

Now perform splits so that $v$ is replaced by exactly $t(v)$ descendants, one
descendant for each block, adjacent precisely to the edges of that block. After
this, each descendant has degree $2$ (for internal blocks) or degree $1$ (for
endpoint blocks). Doing this simultaneously for all vertices makes each trail
$P_i$ become an actual path component in the resulting graph, and different
trails become vertex-disjoint. Therefore the resulting graph lies in $\mathcal{P}$.

The number of splits used “at $v$” is exactly $t(v)-1$. Therefore the total number
of splits is
\[
\sum_{v\in V(G)}(t(v)-1).
\]
If $r=0$ (there are odd-degree vertices), choose $T$ so that each odd-degree vertex
appears as an endpoint exactly once overall; then each $t(v)$ can be taken equal
to $\lceil d_G(v)/2\rceil$, yielding exactly
$\sum_{v}(\lceil d_G(v)/2\rceil-1)$ splits.

If $r=1$ (Eulerian connected), any trail partition must contain at least one trail,
hence it must have two endpoints. This forces exactly one vertex to contribute an
endpoint block (size $1$), which increases $\sum_v t(v)$ by exactly $1$ compared to
the ideal pairing into blocks of size $2$ everywhere. Equivalently, we can realize
all vertices with $t(v)=d_G(v)/2$ except at one chosen vertex where we take
$t(v)=d_G(v)/2+1$. This yields exactly
$\sum_v(\lceil d_G(v)/2\rceil-1)+1$ splits.

Thus, in both cases, the stated number of splits is achievable. Combined with the
lower bounds, this proves optimality for the case where $G$ is a connected graph. Summing over all connected
components gives the general formula with the additional $+r$.
\end{proof}

Note that the construction in the proof of
Theorem~\ref{thm value splitting into union of paths} yields a polynomial-time
algorithm that outputs an explicit splitting sequence transforming any graph
into a disjoint union of paths.

\begin{corollary}
The minimum number of splits required to transform a graph into a disjoint union of paths, and the corresponding split sequence, can be computed in time $O(|E(G)||V(G)|)$.
\end{corollary}

\begin{proof}
The value obtained in Theorem~\ref{thm value splitting into union of paths} can be computed in linear time, and since the trail partition can be computed in time $O(|E(G)||V(G)|)$, the split sequence can therefore be obtaind in (the same) polynomial-time.
\end{proof}

\begin{corollary}
If $G$ is a triangle-free graph, the minimum number of (exclusive or inclusive) splits required to transform $G$ into a unit interval graph can be computed in polynomial time.
\end{corollary}

\begin{proof}
If $G$ is triangle-free and a unit interval graph, then $G$ is a disjoint union of paths.
According to Theorem~\ref{thm value splitting into union of paths}, Exclusive Vertex Splitting into Unit Interval Graphs is solvable in polynomial-time for triangle free graphs. We then conclude the proof by Corollary~\ref{corollary:triangle_free_graphs} (since we have $UIVS(G) = UIVXS(G)$).
\end{proof}

\section{Splitting versus Deletion}

We now consider the (more general) class of chordal graphs. Our proofs have focused
on the (sub)class of interval graphs, so one may suspect that the problem is also
\NP-hard in this case. However, this remains a conjecture for the time being, and
we therefore pose it as an open question. Note that the proof of Theorem~\ref{thm:nphardness} cannot be applied as triangle free chordal graphs correspond exactly to forests, and splitting into forests is already known to be solvable in polynomial time~\cite{firbas2023establishing}. 

We note that modification into a chordal graph is well studied when the editing
operation is edge or vertex deletion. Of course, vertex splitting can present
additional challenges despite its close relationship with deletion. For example,
if splitting $k$ vertices results in a graph belonging to a hereditary class
$\Pi$, then deleting the same set of vertices has the same effect, but the converse
is not necessarily true. In this section, we study the relationship between vertex
splitting and the two deletion operations when the goal is to obtain a chordal
graph. We denote by $ChED(G)$ and $ChVD(G)$ the minimum number of edge deletions and vertex
deletions, respectively, required to transform $G$ into a chordal graph.

We prove that for any graph $G$, $ChVD(G) \leq ChVS(G) \leq ChED(G)$ and that these bounds are not tight.

We also focused our attention on graphs having $\alpha(G) \leq 2$ because they do not contain any induced cycle of length at least $6$.
Therefore to turn such a graph to a chordal graph, we only have to get rid of cycles of length 4 or 5.

\subsection{Vertex Splitting versus Edge Deletion into Chordal Graphs}


We first consider the relationship with Chordal Edge Deletion which is known to be \NP-complete \cite{natanzon2001complexity}.

\begin{theorem}
    For every graph $G$, $ChVS(G) \leq ChED(G)$. Moreover, $ChVS(G)$ is not lower-bounded by $ChED(G)$ even on graphs having $\alpha(G) \leq 2$.
\end{theorem}
\begin{proof}

Let $G$ be a graph.
    Suppose we have a sequence of edge deletions turning the $G$ into a chordal graph.
    For every edge deletion $xy$, replace this operation by splitting $x$ into two copies: one which is only connected to $y$ (and becomes therefore a leaf) and the other which is connected to the other neighbors of $x$.
    This sequence of splits on $G$ has the same length as the edge-deletion sequence, and it  turns the graph into a chordal one.
    In fact, each leaf copy is a simplicial vertex and can be removed, or ignored (no effect on the graph's chordality).
    The remaining graph is the same as the graph obtained after the deletion of the edges.
    Then $ChVS(G) \leq ChED(G)$.

    \vspace{5pt}

    \begin{figure}[!h]
    \centering
        \begin{tikzpicture}
            [ node_style/.style={circle, draw, inner sep=0.05cm}, rotate=90  ]
            \def\fscale{3} 
        	 \node[node_style] (0) at ($\fscale*(-0.5, 0)$) {};
        	 \node[node_style] (1) at ($\fscale*(0, 0)$) {};
        	 \node[node_style] (2) at ($\fscale*(-0.5, -0.5)$) {};
        	 \node[node_style] (3) at ($\fscale*(0, -0.5)$) {};
        	 \node[node_style] (4) at ($\fscale*(-0.5, -1)$) {};
        	 \node[node_style] (5) at ($\fscale*(0, -1)$) {};
        	 \node[node_style] (6) at ($\fscale*(0.5, -0.5)$) {};
        	 \node[node_style] (7) at ($\fscale*(0.5, 0)$) {};
        	 \node[node_style] (8) at ($\fscale*(0.5, -1)$) {};
        	 \node[node_style, label={above:$v$}] (9) at ($\fscale*(0, 1)$) {};
        
        	 \draw (0) to  (1);
        	 \draw (2) to  (3);
        	 \draw (4) to  (5);
        	 \draw (4) to  (3);
        	 \draw (2) to  (5);
        	 \draw (0) to  (3);
        	 \draw (2) to  (1);
        	 \draw (0) to  (5);
        	 \draw (4) to  (1);
        	 \draw (3) to  (6);
        	 \draw (1) to  (6);
        	 \draw (3) to  (7);
        	 \draw (1) to  (7);
        	 \draw (1) to  (8);
        	 \draw (5) to  (6);
        	 \draw (3) to  (8);
        	 \draw (5) to  (8);
        	 \draw (5) to  (7);
        	 \draw (9) to  (0);
        	 \draw (9) to  (2);
        	 \draw (9) to  (4);
        	 \draw (7) to  (9);
        	 \draw (6) to  (9);
        	 \draw (9) to  (8);
        	 \draw (0) to  (2);
        	 \draw (2) to  (4);
        	 \draw (7) to  (6);
        	 \draw (6) to  (8);
        	 \draw (1) to  (3);
        	 \draw (3) to  (5);
        \end{tikzpicture}
        \caption{The graph $G_k$ for $k=3$. This graph satisfies $ChVS(G_k) = 1$ and $ChED(G_k) \geq k$ and $\alpha(G_k) = 2$.}
        \label{fig:chordal_edge_deletion_graph}
    \end{figure}
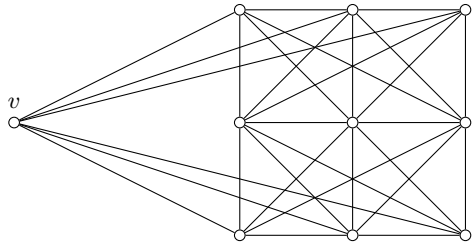

    For any integer $k$, we consider the graph $G_k$ defined as follows: consider the vertices $v,$ $a_1, \ldots, a_k,$ $b_1, \ldots, b_k, c_1, \ldots, c_k$.
    Make $\{a_1, \ldots, a_k\}$, $\{b_1, \ldots, b_k\}$ and $\{c_1, \ldots, c_k\}$ cliques and add the edges $v a_i$, $a_i b_j$, $b_i c_j$ and $c_j v$ for every $i$ and $j \in [k]$.

    With only one split, we can turn this graph into a chordal graph by splitting the vertex $v$ into two copies $v_1$ and $v_2$ such that $v_1$ is adjacent to the vertices $a_1, \ldots, a_k$ and $v_2$ is adjacent to the vertices $c_1, \ldots, c_k$.
    The vertex elimination order is $v_1, v_2, a_1, \ldots, a_k, b_1, \ldots, b_k, c_1, \ldots, c_k$.

    This graph needs at least $k$ edge deletions because there are $k$ edge independent induced $C_4$'s: $v, a_i, b_i, c_i$ for each $i \in [k]$.
    Thus $ChVS(G_k) = 1$ and $ChED(G_k) \geq k$.

\end{proof}

\subsection{Vertex Splitting versus Vertex Deletion into Chordal Graphs}

We now turn to the relationship with Vertex Deletion.
The problem $ChVD$ is known to be \NP-complete \cite{lewis1980node}.

\begin{theorem}
    Let $\Pi$ be a hereditary graph class.
    For any graph $G$, $\Pi VD(G) \leq \Pi VS(G)$.
    In particular, for any graph $G$, $ChVD(G) \leq ChVS(G)$.
\end{theorem}
\begin{proof}
Consider a sequence $\sigma$ of length $k = ChVS(G)$ of splits on $G$, which turns the graph into a graph $G_\sigma$ of $\Pi$.
We denote by $S$ the set of all the vertices of $G$ which are split.
Then $G - S$ is an induced subgraph of $G_\sigma$.
Now since $G_\sigma$ is in $\Pi$ and as $\Pi$ is hereditary, we conclude that $G-S$ is in $\Pi$ and that $\Pi VD(G) \leq \Pi VS(G)$.
\end{proof}

\begin{theorem}
\label{chvd-chvs}
    The parameters $ChVD$ and $ChVS$ are not equivalent.
\end{theorem}

\begin{proof}
    We prove that, for any $k \ge 1$, there exists a graph $G_k$ such that $ChVD(G_k) = 1$ and $ChVS(G_k) \geq k$.
    Let $G_k$ be the graph obtained from a star on $k+1$ vertices by transforming each of its edges into a cycle of order $4$.
    We have $ChVD(G_k) = 1$, as removing the center of the star transforms the graph into a union of $P_3$.
    
    We now prove that at least $k$ splits are required to turn $G_k$ into a chordal graph.
    This graph has $3k + 1$ vertices and $4k$ edges.
    Let $\sigma$ be a sequence of at most $k-1$ splits on $G_k$.
    We denote by $G_\sigma$ the obtained graph.
    This graph has at most $3k+1 + k-1 = 4k$ vertices.
    The number of edges has not decreased (more precisely it is the same if the splits are exclusive).
    Thus $G_\sigma$ has at least $4k$ edges.
    Thus $G_\sigma$ is not acyclic and contains therefore a cycle.
    Consider an induced cycle of $G_\sigma$.
    This cycle cannot be of size $3$ because $G_k$ is triangle-free and splits do not decrease the girth of the graph.
    We conclude that $ChVS(G_k) \geq k$.
\end{proof}

\medskip

We now focus on graphs such that their complement is triangle-free: in other words graphs $G$ such that $\alpha(G) \leq 2$.
In this case, we can prove that Chordal Vertex Deletion is "nearer" to Chordal Vertex Splitting than in the general case.
The next Lemma is used to prove Theorem~\ref{theorem:chordal_vertex_deletion_1}.

\begin{lemma} \label{lemma:chordal_condition}
    For any graph $G$, if $\alpha(G) \leq 2$ and $\chi(\overline{G}) \geq 3$, then $G$ contains an induced $C_4$ or an induced $C_5$.
\end{lemma}

\begin{proof}
    By contradiction, let $G$ be a graph with the smallest number of vertices that does not satisfy Lemma~\ref{lemma:chordal_condition}. 
    
    Since $\chi(\overline G) \ge 3$, there exists $k \ge 1$ and vertices $v_1, \dots v_{2k+1}$, such that $\{v_1, \dots, v_{2k+1}\}$ is a cycle in $\overline G$. Moreover since $\alpha(G) \leq 2$, then $k \ge 2$.

    If $|G| = 5$, necessarily, $k = 2$.
    Since $\alpha(G) \leq 2$, it is not possible that there exists $i$ such that $v_i v_{i+2}$ is not an edge, otherwise there would be a non triangle $v_i, v_{i+1}, v_{i+2}$. We deduce that $G$ is $C_5$ (see Figure~\ref{fig:C5}), a contradiction.

     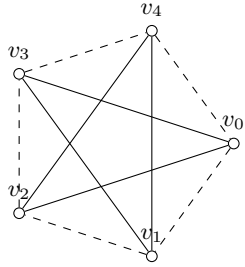
\begin{figure}[!h] 
        \centering
        \begin{tikzpicture}
        [ node_style/.style={circle, draw, inner sep=0.05cm}  ]
        \def\fscale{1.5} 
        \node[node_style, label={above:$v_0$}] (0) at ($\fscale*(0.951, 0)$) {};
    	 \node[node_style, label={above:$v_1$}] (1) at ($\fscale*(0.225, -1)$) {};
    	 \node[node_style, label={above:$v_2$}] (2) at ($\fscale*(-0.951, -0.618)$) {};
    	 \node[node_style, label={above:$v_3$}] (3) at ($\fscale*(-0.951, 0.618)$) {};
    	 \node[node_style, label={above:$v_4$}] (4) at ($\fscale*(0.225, 1)$) {};
    
    	 \draw[dashed] (3) to  (4);
    	 \draw[dashed] (4) to  (0);
    	 \draw[dashed] (0) to  (1);
    	 \draw[dashed] (1) to  (2);
    	 \draw[dashed] (2) to  (3);
    	 \draw (3) to  (1);
    	 \draw (2) to  (4);
    	 \draw (3) to  (0);
    	 \draw (0) to  (2);
    	 \draw (4) to  (1);
    \end{tikzpicture}
        \caption{In the case where $G$ has 5 vertices, we show that $G$ is $C_5$.}
    \label{fig:C5}
    \end{figure}

    Otherwise $|G| \geq 6$ (see Figure~\ref{fig:C4}.
    Consider the subgraph of $G$ induced by the vertices $\{v_1, v_4, v_2, v_5\}$.
    By hypothesis, it cannot be a $C_4$.
    Thus, as we already know that $v_1v_2$ and $v_4v_5$ are non-edges in $G$ and that $v_2v_4$ is an edge in $G$ (otherwise $(v_2, v_3, v_4)$ is a non-triangle ),  there must be one other non-edge among $\{v_1v_4, v_2v_5, v_1 v_5\}$. 
    If $v_1v_5$ is a non-edge, then $(v_1, v_2, v_3,v_4,v_5)$ is an induced subgraph which satisfies the hypothesis of the lemma, which is not possible by the case $|G| = 5$. 
    Therefore $v_1v_4$ ($v_2v_5$ resp.) is a non-edge, and thus the subgraph induced by $\{v_1, v_4, v_5, \dots v_{2k+1}\}$ ($\{v_1, v_2, v_5, v_6, \dots v_{2k+1}\}$ resp.) is a smaller counter example, which contradicts the minimality of $G$.
    A contradiction.

    We conclude that no smallest counter-example exists.

    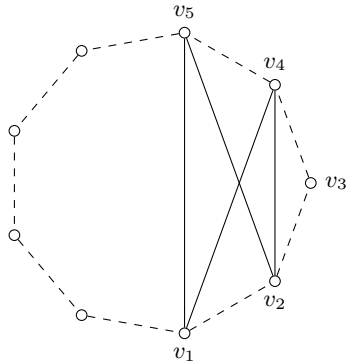
\begin{figure}[!h] 
        \centering
        \begin{tikzpicture}
        [ node_style/.style={circle, draw, inner sep=0.05cm}  ]
        \def\fscale{2} 
    	 \node[node_style, label={right:$v_3$}] (0) at ($\fscale*(0.985, 0)$) {};
    	 \node[node_style, label={below:$v_2$}] (1) at ($\fscale*(0.747, -0.653)$) {};
    	 \node[node_style, label={below:$v_1$}] (2) at ($\fscale*(0.146, -1)$) {};
    	 \node[node_style] (3) at ($\fscale*(-0.538, -0.879)$) {};
    	 \node[node_style] (4) at ($\fscale*(-0.985, -0.347)$) {};
    	 \node[node_style] (5) at ($\fscale*(-0.985, 0.347)$) {};
    	 \node[node_style] (6) at ($\fscale*(-0.538, 0.879)$) {};
    	 \node[node_style, label={above:$v_{5}$}] (7) at ($\fscale*(0.146, 1)$) {};
    	 \node[node_style, label={above:$v_{4}$}] (8) at ($\fscale*(0.747, 0.653)$) {};
    
    	 \draw[dashed] (0) to  (8);
    	 \draw[dashed] (8) to  (7);
    	 \draw[dashed] (7) to  (6);
    	 \draw[dashed] (6) to  (5);
    	 \draw[dashed] (5) to  (4);
    	 \draw[dashed] (4) to  (3);
    	 \draw[dashed] (3) to  (2);
    	 \draw[dashed] (2) to  (1);
    	 \draw[dashed] (1) to  (0);
    	 \draw (8) to  (1);
    	 \draw (1) to  (7);
    	 \draw (7) to  (2);
    	 \draw (2) to  (8);
    \end{tikzpicture}
        \caption{In the case where all the inner edges of the non-cycle are present, there exists an induced $C_4$ in $G$.}
        \label{fig:C4}
    \end{figure}

\end{proof}

\begin{theorem} \label{theorem:chordal_vertex_deletion_1}
    Let $G$ be a graph such that $\alpha(G) \leq 2$.
    Then $ChVD(G) = 1 \iff ChVS(G) =1$.
\end{theorem}

\begin{proof}
    If $ChVS(G) = 1$, then $ChVD(G) = 1$ by deleting the split vertex.

    Suppose now that $ChVD(G) = 1$.
    There exists a vertex $v$, such that $G \setminus \{v\}$ is chordal.
    As $\alpha(G \setminus \{v\}) \leq 2$, we deduce by Lemma~\ref{lemma:chordal_condition} that $\chi( \overline{G \setminus \{v\}}) \leq 2$.
    Thus $G \setminus \{v\}$ can be partitioned into two complete subgraphs $A$ and $B$.

    We define $G'$ by splitting $v$ into $v_1$ and $v_2$ such that $v_1$ is connected to $N(v) \cap A$ and $v_2$ is connected to $N(v) \cap v_2$, the vertices $v_1$ and $v_2$ are simplicial.
    We start an elimination order of $G'$ by deleting these two vertices.
    Then, as $G \setminus \{v\}$ is chordal, we complete this elimination order of $G'$ by concatenating the elimination order of $G \setminus \{v\}$.
    Thus $ChVS(G) = 1$.
\end{proof}

We now prove that the generalization of the previous equivalence is false.

\begin{theorem}
\label{chvd-chvs-alpha}
    The parameters $ChVD$ and $ChVS$ are different among graphs $G$ having $\alpha(G) \leq 2$.
    More precisely, there exists a graph $G$ such that $ChVD(G) < ChVS(G)$ such that $\alpha(G) = 2$.
\end{theorem}

\begin{proof}
We consider the graph $G$ on 9 vertices described on Figure~\ref{fig:chordal_vertex_deletion_alpha_2}.

\begin{figure}
    \centering
        \begin{tikzpicture}
            [ node_style/.style={circle, draw, inner sep=0.05cm}  ]
            \def\fscale{2} 
        	 \node[node_style] (0) at ($\fscale*(0, 1)$) {$0$};
        	 \node[node_style] (1) at ($\fscale*(0.2, 0.6)$) {$1$};
        	 \node[node_style] (2) at ($\fscale*(0.6, 0.2)$) {$2$};
        	 \node[node_style] (3) at ($\fscale*(1, 0)$) {$3$};
        	 \node[node_style] (4) at ($\fscale*(1, -0.4)$) {$4$};
        	 \node[node_style] (5) at ($\fscale*(0.6, -0.6)$) {$5$};
        	 \node[node_style] (6) at ($\fscale*(0.2, -1)$) {$6$};
        	 \node[node_style] (7) at ($\fscale*(-1, 0)$) {$7$};
        	 \node[node_style] (8) at ($\fscale*(-1, -0.4)$) {$8$};
        
        	 \draw (7) to  (0);
        	 \draw (7) to  (1);
        	 \draw (7) to  (3);
        	 \draw (7) to  (4);
        	 \draw (7) to  (6);
        	 \draw (3) to  (4);
        	 \draw (3) to  (5);
        	 \draw (3) to  (6);
        	 \draw (2) to  (5);
        	 \draw (2) to  (6);
        	 \draw (1) to  (6);
        	 \draw (8) to  (0);
        	 \draw (8) to  (1);
        	 \draw (8) to  (2);
              \draw (8) to  (4);
        	 \draw (8) to  (5);
        	 \draw (7) to  (8);
        	 \draw (5) to  (4);
        	 \draw (6) to  (5);
        	 \draw (6) to  (4);
        	 \draw (3) to  (2);
        	 \draw (2) to  (1);
        	 \draw (1) to  (0);
        	 \draw (2) to  (0);
        	 \draw (1) to  (3);
        	 \draw (3) to  (0);
        \end{tikzpicture}
        
    \caption{Example of a graph $G$ which has $ChVD(G) = 2$ and $ChVS(G) > 2$ with $\alpha(G) = 2$. Vertices $0,1,2,3$ and $4,5,6$ are forming cliques.}
    \label{fig:chordal_vertex_deletion_alpha_2}
\end{figure}
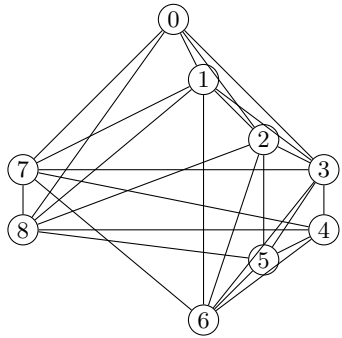

We can check that $\alpha(G) = 2$ as its complement is triangle-free.
Remark that the cycles $(7,0,2,6)$ and $(8,1,3,4)$ are induced.
As these cycles are disjoint, then $ChVD(G) \geq 2$.
The induced subgraph $G \setminus \{7,8\}$ is chordal as there exists an elimination order which is $0,1,2,3,4,5,6$.

We checked with a computer, by brute force, that $ChVS(G) \geq 3$ by trying all possible splits on one vertex of the cycle $(7,0,2,6)$ and on one vertex of the cycle $(8,1,3,4)$.

\end{proof}

\section{Concluding Remarks}

We initiated a study of vertex splitting as a graph
modification operation toward interval graphs. We proved that deciding whether a
graph can be transformed into an interval graph using at most $k$ vertex splits
is \NP-complete, even when the input is restricted to subcubic planar bipartite graphs. 
We also showed that,
on triangle-free graphs, inclusive and exclusive splitting coincide for chordal,
interval, and unit interval (target) graphs.

On the positive side, we presented a
polynomial-time algorithm that computes the minimum number of splits needed to
transform a graph into a disjoint union of paths, together with a corresponding
split sequence. As a consequence, splitting triangle-free graphs into unit
interval graphs is solvable in polynomial time.

We further compared vertex splitting with edge and vertex deletion when the
target graph is chordal. We showed that chordal vertex splitting is always upper-bounded
by chordal edge deletion, and that chordal vertex deletion is always upper-bounded
by chordal vertex splitting. 
At the same time, we presented strong separations between these parameters, even on graphs with independence number at most two. 
These results confirm that splitting-based distance to chordal and interval graph classes behaves differently from deletion-based distance and requires separate analysis.

We conclude with several open questions. What is the complexity of chordal vertex
splitting and unit interval vertex splitting in general? Are there fixed-parameter
tractable or approximation algorithms for interval vertex splitting under natural parameters? 
As for auxiliary parameters, we note that splitting into interval graphs is fixed-parameter tractable when parameterized by the treewidth, and the number of splits of the input graph. This follows immediately from Courcelle's Theorem \cite{courcelle} since  the problem is expressible via monadic second order logic~\cite{eppstein2018planar}. However, the problem remains interesting when parameterized by other auxiliary parameters such as twin-width and modular width, to name a few.


\end{document}